\documentclass{kluwer}
\usepackage{savesym}
\savesymbol{iint}
\savesymbol{iiint}
\usepackage{amsmath}
\restoresymbol{TXF}{iint}
\restoresymbol{TXF}{iiint}
 \usepackage{epsfig}
 \usepackage{graphicx} 
 \usepackage{amsthm}
\def \beq{ \begin{equation} }
\def \eeq{ \end{equation} }

\newcommand{\sgn}{\mathop{\mathrm{sgn}}}

\begin{tiny}\end{tiny}

\def \pd{\partial}

\def \ep{\epsilon}

\def \( {\big( }
\def \) {\big) }

\def \bar{\overline}

\newtheorem{theorem}{Theorem}
\newtheorem{corollary}{Corollary}
\newtheorem{lemma}{Lemma}
\newtheorem{proposition}{Proposition}

\newtheorem{definition}{Definition}

\begin{document}
\begin{article}
\begin{opening}
%% Place the the running title of the paper with 40 letters or less in []
 %% and the title of the paper in { }.
\title{Rosette Central Configurations, Degenerate central configurations and bifurcations}

%% Place all authors' names in [ ] shown as running head;
 %% No more than 40 letters. Leave { } empty
%\author[Jinzhi Lei and Manuele Santoprete]{}
\author{J. \surname{Lei}$^{(1)}$\thanks{Supported by NNSFC (National Natural Science Foundation of China) grant No. 10301006.}}
\author{M. \surname{Santoprete}$^{(2)}$}
\runningauthor{J. Lei and M. Santoprete}
\runningtitle{Rosette Central Configurations}
\institute{$^{(1)}$ Department of Mathematics, University of California, Irvine and Zhou Pei-Yuan Center for Applied Mathematics, Tsinghua University, Beijing, 100084 China, email: jzlei@mail.tsinghua.edu.cn\\
$^{(2)}$ Department of Mathematics, University of California, Irvine, Irvine CA, 92697 USA, email: msantopr@math.uci.edu}

%% Email address is required.
%\email{jzlei@mail.tsinghua.edu.cn, jin\_zhi\_lei@yahoo.com}
%\email{msantopr@math.uci.edu}

%\thanks{The first author was supported by NNSFC (National Natural Science Foundation of China) grant No. 10301006.}

%\begin{document}
%\maketitle

%% Enter the first author's name and address:
%\centerline{\scshape Jinzhi Lei }
%\medskip
%{\footnotesize
% \centerline{Department of Mathematics, University of California, Irvine}
%  \centerline{ and}
%   \centerline{Zhou Pei-Yuan Center for Applied Mathematics, Tsinghua University,}
%   \centerline{Beijing, 100084 China}
%} %% Do not forget to end the {\footnotesize by the sign }
%\medskip
%\centerline{\scshape Manuele Santoprete}
%\medskip
%{\footnotesize
% \centerline{Department of Mathematics, University of California, Irvine}
%  \centerline{ Irvine CA, 92697 USA}
%} %% Do not forget to end the {\footnotesize by the sign }

%\medskip

\begin{abstract}
\noindent In this paper we find a class of new degenerate central configurations and bifurcations in the Newtonian $n$-body problem. In particular we analyze the Rosette central configurations, namely a coplanar configuration where $n$ particles of mass $m_1$ lie at the vertices of a regular $n$-gon, $n$ particles of mass $m_2$ lie at the vertices of another $n$-gon concentric with the first, but rotated of an angle $\pi/n$, and an additional particle of mass $m_0$ lies at the center of mass  of the system. This system admits two mass parameters $\mu=m_0/m_1$ and $\ep=m_2/m_1$. We show that, as $\mu$ varies,  if $n> 3$, there is   a degenerate central configuration and a bifurcation for every $\ep>0$, while if $n=3$ there is a bifurcations only for some values of $\epsilon$.
\end{abstract}
%\subjclass{Primary: 70F15, 70F10}

\end{opening}
%%%%%%%%%%%%%%%%%%%%%%%%%%
\section{Introduction}
%%%%%%%%%%%%%%%%%%%%%%%%%%
In the planar Newtonian $n$-body problem the simplest possible motions are such that the whole system of particles  rotates as a rigid body  about its center of mass. In this case the configuration of the bodies does not change with time. Only some special configurations of point particles are allowed such motions. These configurations are called {\it central configurations}.

Many questions were raised about the set of central configurations.
The main general open problem is the Chazy-Wintner-Smale conjecture: given $n$ positive masses $m_1,\ldots, m_n$ interacting by means of the Newtonian potential, the set of equivalence classes of central configurations is finite. Such conjecture was proved for $n=4$, in the case of equal masses, %by Alain Albouy 
by \inlinecite{Albouy95} and \shortcite{Albouy96} and in the general case by %Marshall Hampton and Rick Moeckel 
\inlinecite{Hampton}.

Chazy believed in a stronger statement: namely that any
equivalence class of central configuration is non-degenerate. This
statement is known to be false: Palmore \cite{Palmore,Palmore76}
showed
 the existence of degenerate central configurations in the planar $n$-body problem with $n\geq 4$.
His example consists of $n-1$ particles lying at the vertices of a regular polygon and one particle at the centroid.  Unfortunately only few examples of degenerate central configurations are known.
In this paper we find a new family of degenerate central configurations that arise from  some highly symmetrical  configurations.

Another interesting problem, that is strictly related to the study of degenerate central configurations, is the  study of bifurcations in the $n$-body problem. The interest in this problem arises because, at a bifurcation, the structure of the phase space changes. Several authors studied bifurcations in the $n$-body problem (see \inlinecite{Sekiguchi} for a list of references), in particular M. Sekiguchi analyzed a highly symmetrical configuration of $2n+1$-bodies.  He considered a rosette configuration, i.e. a planar  configuration where $2n$ particles  of mass $m$ lie at the vertices of two concentric regular $n$-gons, one rotated an angle of $\pi/n$ from the other and another particle of mass $m_0$ lies at the center of the two $n$-gons. He showed that there is a bifurcation in the number of classes of central configurations for any $n\geq 3$.

In this paper we generalize Sekiguchi example and we allow the masses on the two concentric $n$-gons to be different. This considerably complicates the analysis.  Indeed, if one considers two concentric $n$-gons one with particles of mass $m_1$ and the other (rotated of an angle $\pi/n$ from the first) with particles of masses $m_2$ and a mass $m_0$ in the center, one has to deal with two mass parameters $\mu=m_0/m_1$ and $\ep=m_2/m_1$.
In this case we prove that, as $\mu$ varies,  if $n> 3$, there is  a degenerate central configuration and a bifurcation for every $\ep>0$. On the other hand the case $n=3$ is special and, in this case, as $\mu$ is varied, there is a bifurcation for some values of $\epsilon$ but not for others.

This paper is organized as follows. In the next section we introduce the equation of the $n$-body problem. In Section 3 we discuss central configurations. In the following section we introduce the highly symmetrical configurations that are the object of the paper. In Section 5 we present and prove the main results of the paper: the existence, for any $n>3$, of a bifurcation in   the number of classes of central configurations and of a new family of degenerate central configurations. In the last section we analyze the special case where $n=3$.

%%%%%%%%%%%%%%%%%%%%%%%%%%%%
\section{Equations of Motion}
%%%%%%%%%%%%%%%%%%%%%%%%%%%%

The planar $n$-body problem concerns the motion of $n$ particles
with masses $m_i\in{\mathbb R}^+$ and positions $q_i\in{\mathbb
R}^2$, where $i=1,\ldots,n$. The motion is governed by Newton's
law of motion \beq m_i\ddot q_i=\frac{\partial U}{\partial q_i}.
\eeq Where $U(q)$ is the Newtonian potential \beq
U(q)=\sum_{i<j}\frac{m_im_j}{|q_i-q_j|}. \eeq Let
$q=(q_1,\ldots,q_n)\in {\mathbb R}^{2n}$ and
$M=\mathrm{diag}[m_1,m_2,\ldots,m_n]$. Then the equations of
motion can be written as \beq \ddot q=M^{-1}\frac{\partial
U}{\partial q}. \eeq In studying this problem it is natural to
assume that the center of mass of the system is at the origin,
i.e. $m_1q_1+\ldots+m_nq_n=0$, and that the configuration avoids
the set $\Delta=\{q:q_i=q_j~ \mbox{for some }~q_i\neq q_j\}$.

\section{Central Configurations}
\begin{definition}
A configuration $q\in {\mathbb R}^2\setminus \Delta$ is called a central configuration if there is some constant $\lambda$ such that
\[
M^{-1}\frac{\partial U}{\partial q}=\lambda q.
\]
\end{definition}

Central configurations, as it was shown by Smale (see \cite{Abraham,Smale70}), can be viewed as rest points of a certain gradient flow. Introduce a metric in ${\mathbb R}^{2n}$ such that $\langle q,q \rangle=q^TMq$ and let
\[
S=\{q:\langle q,q\rangle=1, m_1q_1+\ldots +m_nq_n=0\}
\]
denote the unit sphere $S^{2n-3}$ with respect to this metric in
the subspace where the center of mass is at the origin.  The
scalar product $I=\langle q,q\rangle$ is called moment of inertia.
Let $S^*=S\setminus\Delta$. The vector field
$X=M^{-1}\frac{\partial U}{\partial q}+\lambda q$ where
$\lambda=U(q)$ is the gradient of $U_S$, the restriction of  $U$
to the unit sphere $S$ with respect to the metric $\langle
\cdot,\cdot\rangle$. This is because $X$ is tangent to $S$ , it
has rest points at exactly the central configurations with
$\langle q,q\rangle=1$ and $\langle X(q),v\rangle=DU(q)v$ for
every $q\in S$ and $v\in T_qS$. Furthermore the rest points of $X$
are exactly the central configurations in $S$. Note that, since
the Newtonian potential is an homogeneous function, any central
configuration is homothetic to one in $S$. Therefore the problem
of finding central configurations is essentially that of finding
rest points of the gradient flow of $U_S$ or, equivalently,
finding the critical points of $U_S$.

The gradient flow preserves  some sets of configurations with symmetry. In this paper we study one of such sets of configurations with symmetry.

We denote by $C_n$ the set of
central configuration of the  $n$-body problem. We say
that two relative equilibria in $S^*$ are equivalent (and
belong to the same equivalence class) %if they can be made
%congruent by the induced $S^1$ action on $S^*$, that is
if one is obtained from the other by a rotation and an homothety. The set $\tilde C_n$
is the set of equivalence classes of central configurations.

Clearly $I$ and $\Delta$ are invariant under the action of $S^1$.
Thus, we can conclude that $S^*$ is diffeomorphic to the $(2n-3)$-dimensional sphere $S^{2n-3}$ (it is actually an ellipsoid $E^{2n-3}$) with all the points $\Delta $ removed, that is
\[S^*=E^{2n-3}\setminus (E^{2n-3}\cap\Delta)\approx S^{2n-3}\setminus (S^{2n-3}\cap\Delta).\]
 Since $U_{S}$ is invariant under the action of $S^1$ it defines a map $\tilde U_{S}:S^*/S^1\rightarrow {\mathbb R}$.
 If we let $\pi:S^*\rightarrow S^*/S^1$ denote the canonical projection, $\tilde\Delta=\pi(E^{2n-3}\cap\Delta)$, and recalling that $E^{2n-3}/S^1\approx S^{2n-3}/S^1\approx {\mathbb C}P^{n-2}$, complex projective space, we are led to the investigation of the critical points of $\tilde U_{S}:{\mathbb C}P^{n-2}\setminus \tilde\Delta \rightarrow {\mathbb R}$.

Consequently one can show that the set of equivalence classes of central configurations  is given by the set of critcal points of the map $\tilde U_{S}:{\mathbb C}P^{n-2}\setminus \tilde\Delta\rightarrow {\mathbb R}$. More precisely we have the following result of Smale (see \cite{Abraham,Smale70,Smale71})

\begin{proposition}
For any $n\geq 2$ and any choices of the masses in the planar $n$-body the set of equivalence classes of central configurations is diffeomorphic to the set of critical points of the map $\tilde U_{S}:{\mathbb C}P^{n-2}\setminus\tilde\Delta\rightarrow {\mathbb R}$.
\end{proposition}

Let  $q$ be a critical point of $\tilde U_S$. A critical point of $\tilde U_S$ is degenerate provided that the hessian $D^2\tilde U_S(q)$ has a nontrivial nullspace. We have the following definition
\begin{definition}
An equivalence  class of central configurations is degenerate (nondegenerate) provided that the corresponding critical point $q$ of $\tilde U_S$ is degenerate (nondegenerate).
\end{definition}
%%%%%%%%%%%%%%%%%%%%%%%%%%%%%%%%%%%%%%%%%%%%
\section{Symmetrical Configurations }
%%%%%%%%%%%%%%%%%%%%%%%%%%%%%%%%%%%%%%%%%%%%%%%%

%---------------------------------------------------------
\begin{figure}[t]
\begin{center}
\resizebox{!}{8cm}{\includegraphics{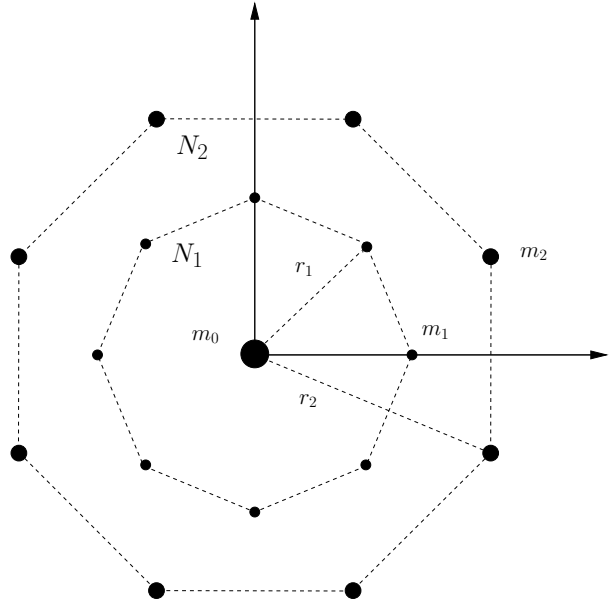}}
\end{center}
\caption{Rosette configuration for $n=6$}
\label{rosette}
\end{figure}
%----------------------------------------------------

Consider the set $\Sigma$ of all the configuration in ${\mathbb R}^2$ consisting of two concentric regular $n$-gons, one rotated of an angle $\pi/n$ from the other,  with a mass in their common center of symmetry
 (see Figure \ref{rosette}).

Let $m_0$, $m_1$ and $m_2$ be the masses in the center of mass, on the  $n$-gon $N_1$ and on the $n$-gon $N_2$ respectively.
 Then it follows from the symmetry of the configuration that the gradient of $\tilde U_S$ is tangent to $\tilde \Sigma$ (where $\tilde \Sigma=\pi(E^{2n-3}\cap\Sigma$)). Thus to find  equivalence classes of central configurations in $\tilde\Sigma$ it is sufficient to study the
critical points of $\tilde U_S|_{\tilde \Sigma}$. Since $\tilde
\Sigma$ is one dimensional, only one parameter is needed to
describe such symmetric  configuration. This is a great
simplification. Figure \ref{rosette} shows two parameters
$(r_1,r_2)$ which can be used to describe such a configuration.

The potential in these coordinates is
\[
U(q)=(nm_1)^2 \bar U(r_1,r_2)
\]
where
\[
\bar U(r_1,r_2)= \frac \mu n \left(\frac{1}{r_1}+
\frac{\ep}{r_2}\right)+k_n\left(\frac{1}{r_1}
+\frac{\ep^2}{r_2}\right)+\frac{1}{n}\sum_{k=1}^{n}
\frac{\ep}{\sqrt{r_1^2+r_2^2-2r_1r_2\cos\phi_k}}\] with,
$\mu=m_0/m_1$, $\ep=m_2/m_1$, $\phi_k=(2k-1)\pi/n$ and
\[k_n=\frac{1}{4n}\sum_{k=1}^{n-1} \csc\frac \pi n k\]
is the potential of a regular $n$-gon of unit size and unit
masses. The last term in $\bar{U}$ is the moment of inertia
\[I(q)=\langle q,q \rangle=\bar I(r_1,r_2)=m_1n(r_1^2+\ep r_2^2).\]
The  central configurations are the solutions of the equation $\nabla\tilde U_S|_{\tilde \Sigma}=0 $ or $\nabla U=\frac{\lambda}{2} \nabla I$ (with $\lambda=U$), that in this case can be written as
\beq
\begin{split}
\frac{\pd \bar U}{\pd r_1}=& m_1 n \lambda r_1\\
\frac{\pd \bar U}{\pd r_2}=& m_1 n \lambda \ep r_2.
\end{split}
\eeq
Solving the equations above for $\lambda$ one gets
\beq
\frac{1}{r_2^3}F(x)=0
\eeq
where
\beq\label{centralca}
F(x,\ep,\mu)= \frac \mu n(1-x^3) +k_n(\ep-x^3)+\frac{x^3}{n} \sum_{k=1}^n\frac{(1-\ep)-\frac 1 x (1-\ep x^2)\cos\phi_k}{(1+x^2-2x\cos\phi_k)^{3/2}}
\eeq
and $x=r_2/r_1$.
The equation for the central configuration above depends only on one parameter and  is invariant under the transformation $(x,\ep,\mu)\rightarrow (\frac 1 x, \frac 1 \ep, \mu\ep)$. Thus  it suffices to study the central configurations with $0<\ep\leq 1$.

The  case $\ep=1$ was studied  in detail by   \inlinecite{Sekiguchi} that proved the following
\begin{theorem}\label{theosekiguchi}
If $n=2$ the number of central configurations is one for any value of $\mu$. If $n\geq 3$ the number of central configurations is three for $\mu<\mu_c(n)$ and one for $\mu\geq \mu_c(n)$, where
\[
\mu_c(n)=\frac {1}{12} \sum_{k=1}^n \frac{\cos\phi_k}{\sin^3(\phi_k/2)}-nk_n
\]
\end{theorem}

It is therefore sufficient to analyze the problem with $0<\ep<1$.

\begin{proposition}
For every $\mu>0$ and $\ep>0$ there is at least one rosette central configuration.
\end{proposition}
\begin{proof}
Since $\lim_{x \to 0}F(x)=\left(\frac \mu n +k_n\right)>0$ and $\lim_{x\to \infty} F(x)=-\infty<0$, by the intermediate value theorem, the equation $F(x)=0$ has at least one solution.
\end{proof}

When $n=2$ it can be shown that, for every value of $\epsilon$, there is only one class of central configurations and no bifurcation occur, or more precisely we have
the following

\begin{proposition}
If $n=2$ for every $\mu>0$ and $\ep>0$ there is only one  rosette central configurations.
\end{proposition}
\begin{proof}
In this case \beq F(x,\ep,\mu)=\frac \mu 2 (1-x^3)+\frac 1 8
(\ep-x^3)+\frac{x^3(1-\ep)}{(1+x^2)^{3/2}} \eeq
 $\lim_{x\to 0}F(x)=\frac \mu 2+\frac{\ep}{8}$ and $\lim_{x\to \infty}F(x)=-\infty$ so $F(x)=0$ has at least one solution.
We need only to prove the statement for $0<\ep\leq 1$.
If $\ep=1$ $F(x)$ is a monotonically decreasing and the statement follows. If $0<\ep<1$ consider
 \beq
 F'(x)=3x^2\left(-\left(\frac \mu 2 +\frac 1 8\right ) +\frac{(1-\ep)}{(1+x^2)^{5/2}}\right).
 \eeq
Clearly one solution of $F'(x)=0$ is $x=0$. The other solutions can be found studying the equation $\eta(x)=\frac \mu 2+\frac 1 8$, where
\[
\eta(x)= \frac{(1-\ep)}{(1+x^2)^{5/2}}
\]
is a monotonically decreasing function and $\eta(0)=(1-\ep)$. The equation $\eta(x)=\frac \mu 2+\frac 1 8$ has no solutions if $\mu\geq\frac 7 4$  or $\mu<\frac 7 4$ and $\ep\in(\frac 7 8-\frac \mu 2, 1)$. It has one solution $x^*$ if $\mu<\frac 7 4$ and $\ep\in (0,\frac 7 8-\frac \mu 2]$. Consequently if $\mu\geq\frac 7 4$  or $\mu<\frac 7 4$ and  $\ep\in(\frac 7 8-\frac \mu 2, 1)$  $F'(x)$ is always negative, $F(x)$ monotonically decreasing and $F(x)=0$ has only one solution. On the other hand, if $\mu<\frac 7 4$ and  $\ep\in (0,\frac 7 8-\frac \mu 2]$, $F'(x)$ is  positive  for $x\in(0,x^*)$ and negative for $x\in(x^*,\infty)$. Thus $F(x)$ is increasing for $x\in(0,x^*)$, decreasing for $x\in(x^*,\infty)$ and $F(x)=0$ has one solution since $F(0)>0$.
\end{proof}

%%%%%%%%%%%%%%%%%%%%%%%%%
\section{Bifurcations and degenerate central configurations for $n>3$}
%%%%%%%%%%%%%%%%%%%%%%%%
In this section we consider the rosette central configurations for $n>3$. The main result is the existence of a bifurcations for every value of $\ep$ as the parameter $\mu$ increases. The case  $n=3$ is studied in the next section. More precisely we prove the following

\begin{theorem}\label{mainth}
For any $n>3$ and $\ep>0$  there is at least one  value $\mu_0$  corresponding to a bifurcation in the number of equivalence classes of rosette central configurations as the parameter $\mu>0$ increases.
\end{theorem}

An important consequence of the existence of a bifurcation is the existence of a degenerate equivalence class of rosette central configurations
\begin{corollary}
For any $n>3$ and $\ep>0$ there is at least one  value  $\mu_0$ of $\mu$ for which there is a degenerate equivalence class of rosette central configuration.
\end{corollary}
\begin{proof} The proof is by contradiction.
Consider the potential $\tilde U_S(q;\mu)$ for the configuration
under discussion in this paper, where we put into evidence the
dependence on the mass $\mu$. Let $q_1^0,\ldots,q_l^0$ be the
critical points of $\tilde U_S$ for $\mu=\mu_0$, where $\mu_0$ is
the bifurcation value. Assume that the class of central
configurations is nondegenerate for every $q_l^0$. This means that
$D^2\tilde U_S(q^0_l;\mu)$ has bounded inverse. But then by the
implicit function theorem, there exist a neighborhood $B$ of
$\mu_0$ and unique functions $\{q_l(\mu)\}_{l=1}^n$ defined in
$B$, such that $q_l(\mu_0)=q_l^0$ and $D\tilde
U_S(q_l(\mu);\mu)=0$. This contradicts the assumption that $\mu_0$
is a bifurcation value.
\end{proof}

The proof of  Theorem \ref{mainth} requires several preparations. The reminder of this section is devoted to such preparations  and to the proof of Theorem \ref{mainth}

First of all observe that the  central configurations, when $\ep\neq 1$ can also be viewed as the solutions of $h(x,\ep)=\mu$ where

\[h(x, \epsilon) = - n\, k_n\, \dfrac{\epsilon - x^3}{1 - x^3} - \dfrac{x^2}{(1-x^3)}\,
\sum_{k = 1}^n\dfrac{x\,(1-\epsilon) -
(1-\epsilon\,x^2)\,\cos\phi_k}{(1+x^2 - 2\,x\,\cos\phi_k)^{3/2}}\]
Hereinafter, we say $x$ to be a rosette central
configuration if $x$ is solution of the equation $\mu = h(x,\epsilon)$.
Let $u_k=\cos\phi_k$ then
$$h(x,\epsilon)=h_0(x)+(1-\epsilon) h_1(x)$$
 where
\[
h_0(x)=-nk_n+\frac{x^2(1-x^2)}{(1-x^3)}\sum_{k=1}^{n}\frac{u_k}{(1+x^2-2xu_k)^{3/2}}
\]
and
\[
h_1(x)=\frac{nk_n}{(1-x^3)}-\frac{x^3}{1-x^3}\sum_{k=1}^n\frac{(1-xu_k)}{(1+x^2-2xu_k)^{3/2}}.
\]

%%%%%%%%%%%%%%%%%%%%%%%%%%%%%%%%
\subsection{The case $x>1$, $\ep\in(0,1)$}
%%%%%%%%%%%%%%%%%%%%%%%%%%%%%%%%

We now study the number of central configurations for $x>1$ and $\ep\in(0,1)$.
It is easy to show that, for any $\mu>0$, there is at least one rosette central configuration with $x>1$.
This follows from the limits
\[\lim_{x\to 1^+}h(x,\epsilon) = \infty,\ \ \ \ \lim_{x\to \infty} h(x,\ep) = -n\,k_n < 0\]
and an application of the Intermediate Value Theorem. The first limit is
\[\lim_{x\to 1^+}h(x,\epsilon) =\infty\times \sgn\left( (\ep-1)A_n\right)\]
where
 \beq
A_n=nk_n-\frac 1 4 \sum_{k=1}^n \csc\left( \frac{\phi_k}{2}
\right) \eeq and $\sgn\left( (\ep-1)A_n\right)=1$ since
$\ep-1<0$ and $A_n<0$ by the following Lemma

\begin{lemma}\label{lemmaA_n}
For all $n \geq 2$,
\[A_n<0.\]
\end{lemma}
\begin{proof}
Clearly
\beq
A_n=\frac{1}{4}\left(\sum_{k=1}^{n-1}\csc \frac{k\,\pi}{n} -
\sum_{k=1}^{n}\csc\left(\frac{k\,\pi}{n} - \frac{\pi}{2\,n}\right)\right)
\eeq
therefore  when $n$ is even one has
\beq\begin{split}
 \sum_{k=1}^{n-1}\csc&
\frac{k\,\pi}{n} - \sum_{k=1}^{n}\csc\left(\frac{k\,\pi}{n} - \frac{\pi}{2\,n}\right)\\
=& \sum_{k = 1}^{n/2}\left(\csc \frac{k \pi}{n} -
\csc\left(\frac{k\,\pi}{n} - \frac{\pi}{2\,n}\right)\right) -
\csc\left(\frac{\pi}{2} + \frac{\pi}{2\,n}\right)\\
%&&{}
&+\sum_{k = \frac{n}{2} + 1}^{n-1}\left(\csc \frac{k\,\pi}{n} -
\csc\left(\frac{k\,\pi}{n} + \frac{\pi}{2\,n}\right)\right)<0\end{split}
\eeq
while when $n$ is odd
\beq\begin{split}
\sum_{k=1}^{n-1}\csc& \frac{k\,\pi}{n} -
\sum_{k=1}^{n}\csc\left(\frac{k\,\pi}{n}-\frac{\pi}{2\,n}\right)\\
=& \sum_{k=1}^{(n-1)/2}\left(\csc\frac{k\,\pi}{n} -
\csc\left(\frac{k\,\pi}{n}
- \frac{\pi}{2\,n}\right)\right) - \csc \frac{\pi}{2}\\
&+ \sum_{k = (n+1)/2}^{n-1}\left(\csc \frac{k\,\pi}{n} - \csc
\left(\frac{k\,\pi}{n} + \frac{\pi}{2\,n}\right)\right)< 0.
\end{split}\eeq
\end{proof}

 We now want to  show that
when $\mu$ is large enough, for every $\epsilon \in (0,1)$, there
is exactly one rosette central configuration with $x > 1$, i.e.,
we prove the following

\begin{proposition}\label{propx>1}
For every $\epsilon \in (0,1)$ there exists a $\hat{\mu}$ such
that for every $\mu>\hat{\mu}$ there is one and only one rosette
central configuration
\end{proposition}
\begin{proof}
Observe that one can write \beq h(x,\ep)= -\frac {(1-\ep)
A_n}{3(x-1)}+O((x-1)^0). \eeq Therefore there exist
$\hat{\mu}_0>0$ and $\delta > 0$ such that for any $\mu >
\hat{\mu}_0$ the equation $h(x,\ep)=\mu$ has a unique solution in
$(1,1+\delta)$. Moreover the function $h(x,\ep)$ has a maximum
value $\hat{\mu}_1$ in $[1+\delta,\infty)$, since $\lim_{x\to
\infty}h(x,\ep)=-nk_n$. Let
\[\hat{\mu}=\max(\hat{\mu}_0,\hat{\mu}_1)\]
then, if $\mu> \hat{\mu}$, the equation $h(x,\ep)=\mu$ has a
unique solution.
\end{proof}

%%%%%%%%%%%%%%%%%%%%%%%%%%%%%%%%
\subsection{The case $0<x<1$, $\ep\in(0,1)$}
%%%%%%%%%%%%%%%%%%%%%%%%%%%%%%%%
We now study the number of central configurations for $x<1$ and $\ep\in(0,1)$.
In particular we show that
\begin{proposition}\label{propx<1}
For any $n>3$  and $\epsilon\in (0,1)$,
\begin{enumerate}
\item there is a $\mu_n^*>0$  such that for every $0<\mu<\mu_n^*$
there  are at least two rosette central configurations, with $x\in
(0,1)$ \item there is a $\check{\mu}\geq \mu_n^*   $ such that for
every $\mu> \check{\mu}$ there are no rosette central configurations
with $x\in (0,1)$ .
\end{enumerate}
\end{proposition}
\begin{proof}(a)
Observe that, if $\epsilon \in (0,1)$,
$$h(0,\epsilon) = - n\,k_n\,\epsilon < 0, \quad \lim_{x\to 1^-}h(x, \epsilon) = - \infty\times\sgn\left((\ep-1)A_n)\right)=-\infty,$$
where the limit follows from Lemma \ref{lemmaA_n}. If there exists
$x_n^*\in (0, 1)$ such that $h_1(x_n^*) = 0$ and $h_0(x_n^*) > 0$,
then by the Intermediate Value Theorem, $\mu=h(x,\ep)$ has at
least two solutions for every $0<\mu<\mu_n^*=h_0(x_n^*)$. To
complete the proof it  is necessary to show the existence of
$x_n^*$. The existence of $x_n^*$ will be proved in Lemma
\ref{lemmamain}.

(b) Since $h(0,\epsilon)  < 0$ and $\lim_{x\to 1^-}h(x, \epsilon)
= -\infty<0$ the function $h(x,\ep)$ has a maximum value in
$[0,1]$. Let $\check{\mu}$ be such maximum. Then the equation
$\mu=h(x,\ep)$ has no solutions for $x\in(0,1)$ if
$\mu>\check{\mu}$.
\end{proof}
To complete the proof of the proposition above, and prove Lemma \ref{lemmamain} we  need the following technical result
\begin{lemma} Let
\beq\label{eqrick} k_n^-=\frac{1}{4\pi}\ln\left(\frac{1+\cos\frac{\pi}{n}}{1-\cos\frac{\pi}{n}} \right)+\frac{1}{4n\sin\frac{\pi}{n}}\eeq
then for all $n\geq 3$
\[k_n>k_n^-.\] Moreover $k_n^-$ is  monotonically increasing with $n$.
\end{lemma}
\begin{proof}\label{Rick}
The sum:
\[\sum_{k=1}^{n-1}\csc\frac{\pi k}{n}\]
can be estimated using the trapezoidal rule. Since $g(u)=\csc\frac{\pi u}{n}$ is convex on $[1,n-1]$ the trapezoidal rule gives an upper bound for the integral over $[1,n-1]$:
\[\int_1^{n-1}g(x)~dx<\frac 1 2 g(1)+g(2)+\ldots+g(n-2)+\frac{1}{2}g(n-1).\]
This gives the formula for $k_n^-$. Moreover $k_n^-$ is monotonically increasing since the derivative of the function, obtained replacing $\frac \pi n$ in $k_n^-$ with the continuous variable $u$,  is negative.
\end{proof}
We can finally prove the following
\begin{lemma}\label{lemmamain}
For any $n>3$, there exists a  $x_n^*\in (\frac{54}{100},1)$ such that $h_1(x_n^*) = 0$ and $h_0(x_n^*)>0$.
\end{lemma}

\begin{proof}

Verifying these conclusions numerically for small $n$ is trivial
by common mathematical software (for example,
\verb"Mathematica"\footnote{http://www.wolfram.com/} or
\verb"Matlab"\footnote{http://www.mathworks.com/}). Numerical
results for $4\leq n\leq 106$ are given at Figure \ref{fig:3} (The
solutions $x_n^*$ are found numerically through the function
\verb"FindRoot" provided by \texttt{Mathematica}). The proof for
$n\geq 107$ is given below.

The proof will be completed by showing firstly $h_1(x)=0$ has
solution $x_n^*\in (\frac{54}{100},1)$ and secondly $h_0(x) > 0$
for any $x\in (\frac{54}{100},1)$ such that $h_1(x) = 0$.

%-----------------------------------------------
\begin{figure}[t]
\centering{
  \includegraphics[width = 12 cm]{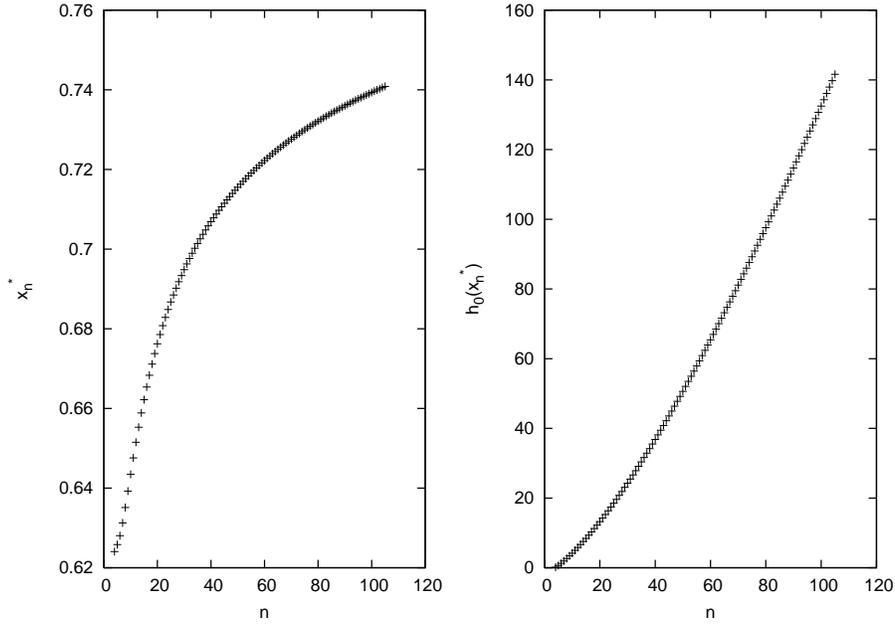}\\
  \caption{The root $x^*_n$ such that $h_1(x_n^*) = 0$ and corresponding $h_0(x_n^*)$ for $4\leq n \leq 106$.}\label{fig:3}
}
\end{figure}
%----------------------------------------------

1. We first show that for  any $n\geq 107$, the equation
$h_1(x)=0$ has at least one solution $x_n^*\in(\frac{54}{100},1)$.
To this end, it is sufficient to show that $h_1(\frac{54}{100}) >
0$ and $h_1(1)<0$. Equivalentlly, let
$$\tilde{h}_1(x) = (1-x^3)\,h_1(x) = nk_n-x^3\sum_{k=1}^n \dfrac{1- x u_k}{(1+x^2 - 2 x u_k)^{3/2}}$$
we will show that $\tilde{h}_1(\frac{54}{100}) > 0$ and $\tilde{h}_1(1) < 0$.

When $u\in [-1,1]$ and $x\in [0,1]$, we have
$$\frac{1-xu}{(1+x^2-2xu)^{3/2}} < \frac{1}{(1-x)^2}.$$
and therewith
\[%\begin{split}
\tilde{h}_1(x)> n\left(k_n-\frac{x^3}{(1-x)^2}\right)
%\end{split}
\]
Thus, when $n \geq 107$
$$\tilde{h}_1(\frac{54}{100}) > n\,\left(k_n - \frac{ \left(\frac{54}{100}\right)^3}{(1-\frac{54}{100})^2}\right) > n\,(k_n^- - \frac{75}{100}) \geq n\,(k_{107}^- - \frac{75}{100}) > 0$$
where $k_{107}^- = 0.7514096544$ was computed using Lemma
\ref{Rick} and $k_n^->k_{107}^-$ since $k_n^-$ is monotonically
increasing with $n$.

A simple computation shows that \beq
\tilde{h}_1(1)=A_n=\frac{1}{4}\left(\sum_{k=1}^{n-1}\csc
\frac{k\,\pi}{n} - \sum_{k=1}^{n}\csc(\frac{k\,\pi}{n} -
\frac{\pi}{2\,n})\right), \eeq and thus, by Lemma \ref{lemmaA_n},
$\tilde{h}_1(1) < 0$ for any $n$. Hence, we conclude that for any
$n\geq 3$, there exist $x_n^*\in (\frac{54}{100},1)$, such that
$h_1(x_n^*) = 0$.

2. We now show that for any $n\geq 107$ and  $x_n^*\in
(\frac{54}{100},1)$ such that $h_1(x_n^*)=0$, $h_0(x_n^*) > 0$.

Let
$$h_2(x) = \frac{1-x^3}{x^2\,(1-x^5)}\,(h_0(x) + (1-x^3)\,h_1(x))$$
then
$$h_0(x_n^*) = \frac{{x_n^*}^3}{R_1(x_n^*)}\,h_2(x_n^*)$$
where
$$R_1(x) = \frac{x\,(1+x+x^2)}{1+x+x^2 + x^3 + x^4}$$
Thus, it is sufficient to prove that $h_2(x) > 0$ for any $x\in
(\frac{54}{100},1)$. To this end, introduce the notations
\begin{eqnarray*}
R_2(x) &=& 0.15 \,R_1(x) + 0.85\\
g(x,u) &=&\dfrac{u - R_1(x)}{(1 + x^2 - 2\,x\, u)^{3/2}}
\end{eqnarray*}
then
$$h_2(x) = \sum_{k = 1}^ng(x,u_k).$$
It is easy to have  $$0 < R_1(x) < R_2(x) < 1,\ \ \forall
x\in (0,1)$$ Thus, grouping the subscripts $k$ in the summation as following
\begin{eqnarray*}
J_1 &=& \{k\ | 1\leq k\leq n,\ u_k < R_1(x)\}\\
J_2 &=& \{k\ | \ 1\leq k\leq n,\ u_k \geq R_2(x)\}
\end{eqnarray*}
we have
$$h_2(x) \geq \sum_{k\in J_1}g(x,u_k) + \sum_{k\in J_2}g(x,u_k).$$
Now, the function $g(x,u)$ of $u\in [0,1]$ (with given $x\in
(\frac{54}{100},1)$) has minimum at $$u = u_-(x) =
\dfrac{3\,x\,R_1(x) - 1 - x^2}{x}$$ and is increasing when $R_2(x)
< u < 1$. Thus, we have when $k\in J_1$,
$$0 > g(x,u_k) > g(x,u_-(x))$$
and when $k\in J_2$,
$$g(x,u_k) \geq g(x,R_2(x)) > 0$$
The number of elements in $J_1$ and $J_2$ are respectively
\begin{eqnarray*}
N(J_1) = \left\lfloor n\,(1 - \frac{\arccos
R_1(x)}{\pi})\right\rfloor < \frac{n}{\pi}\,(\pi - \arccos R_1(x))
\\
N(J_2)  = 2\,\left\lfloor\frac{n}{2\,\pi}\arccos R_2(x) +
\frac{1}{2}\right\rfloor \geq \frac{n}{\pi}\arccos R_2(x) - 1
\end{eqnarray*}
Therefore, we have
\begin{eqnarray*}
h_2(x)&>&\sum_{k\in J_1}g(x,u_-(x))+ \sum_{k\in J_2}g(x,R_2(x))\\
&=&N(J_1)\,g(x,u_-(x)) + N(J_2)\,g(x,R_2(x))\\
&>&\frac{n}{\pi}\left((\pi - \arccos R_1(x))\,g(x,u_-(x)) +
\arccos R_2(x)\, g(x,R_2(x))\right) - g(x, R_2(x))\\
&:=&h_3(x;n)
\end{eqnarray*}

Now, we only need to verify $h_3(x;n)>0$ for any $n\geq 107$ and
$x\in (\frac{54}{100},1)$. It is evident that $h_3(x,n)$ is
increasing with respect to $n$, and thus $h_3(x;107) > 0$, which
is shown at Figure \ref{fig:1}, is enough to complete the proof.

The Lemma has been proved.
\end{proof}
%-----------------------------------------------------------------
\begin{figure}[t!]
\centering{
  \includegraphics[width = 8cm]{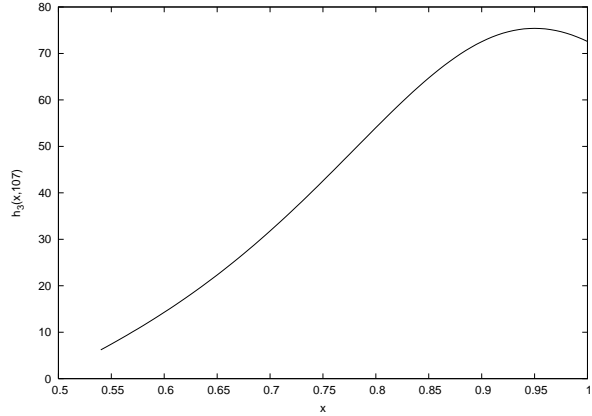}\\
  \caption{The function $h_3(x;n)$ with $n = 107$ and $x\in (0.54,1)$}
  \label{fig:1}}
\end{figure}
%----------------------------------------------------------------

%\textbf{Note}:\ \ The result of Theorem \ref{th:1} is not true
%when $n = 3$. The case $n = 3$ will be discussed latter.

%%%%%%%%%%%%%%%%%%%%%%%%%%%%%%%
\subsection{Proof of Theorem \ref{mainth}}
%%%%%%%%%%%%%%%%%%%%%%%%%%%%%%%%%%
With all the preparations above we are now well on our way to
proving Theorem \ref{mainth}.

On one hand, using Proposition \ref{propx>1} and \ref{propx<1}, we
have, for every $\ep\in (0,1)$, that if
$\mu>\max(\hat{\mu},\check{\mu})$ the equation $\mu=h(x,\ep)$ has
a unique solution.

On the other hand, by Proposition \ref{propx<1}, we have, for
every $\ep\in (0,1)$,  that if $\mu<\mu^*_n$ the equation
$\mu=h(x,\ep)$ has at least two solutions for $x\in(0,1)$ and at
least one solution for $x>1$. Moreover if $\ep\neq 1$,  $x=1$ is
not a solution of $\mu=h(x,\ep)$. Thus the number of rosette
central configurations changes as the parameter $\mu$ increases.
The fact that this result  holds for every $\ep>0$ follows from
Theorem \ref{theosekiguchi} and the invariance under the
transformation $(x,\ep,\mu)\rightarrow (\frac 1 x, \frac 1 \ep,
\mu\ep)$.

This concludes the proof of Theorem \ref{mainth}.

%%%%%%%%%%%%%%%%%%%%%%%%%%
\section{The case $n=3$}
%%%%%%%%%%%%%%%%%%%%%%%%%%

The case $n=3$ is special, indeed for $n=3$ the proof of Lemma
\ref{lemmamain} fails. This is because
$x_3^*=0.617364>\frac{54}{100}$ but  $h(x_3^*)=-0.188154<0$.

In this case we  study  numerically the maximum
$h_{\mathrm{max}}(\ep)$ of the function $h(x,\ep)$ (as a function
of $x$) on the interval $(0,1)$. Figure \ref{fig:sub}(a) depicts
$h_{\mathrm{max}}(\ep)$ for $\ep\in(0,1)$. Figure \ref{fig:sub}(b)
shows a magnification of Figure \ref{fig:sub}(a) near $\ep=0$
making apparent that, near $\ep=0$, $h_{\mathrm{max}}(\ep)>0$.
From Figure \ref{fig:sub}(a)-(b) it is apparent that
$h_{\mathrm{max}}(\ep)$ is always negative except when $\ep$ is
close to 0 or  to 1. More precisely we find that
$h_{\mathrm{max}}(\ep)>0$ for $\ep\in(0,\ep_1)$ and
$\ep\in(\ep_2,1)$ while $h_{\mathrm{max}}(\ep)<0$ for
$\ep\in(\ep_1,\ep_2)$, where $\ep_1=0.00076760883$ and
$\ep_2=0.97198893434$. On the other hand it can be proved that
$h(x,\ep)$ is a monotone decreasing function with respect to $x$
for $x\in(1,\infty)$ and $\ep\in (0,1)$. In fact, when $n = 3$, we
have
$$h(x,\ep) = h_0(x) + (1-\ep) h_1(x)$$
where
\begin{eqnarray*}
h_0(x)&=&-\frac{\sqrt{3}}{3} +
\frac{x^2\,(1+x)}{1+x+x^2}\left(\frac{1}{(1-x+x^2)^{3/2}} -
\frac{1}{(1+x)^3}\right)\\
h_1(x)&=&\frac{\sqrt{3}}{3 (1-x^3)} -
\frac{x^3}{1-x^3}\left(\frac{1}{(1+x)^2} +
\frac{2-x}{(1-x+x^2)^{3/2}}\right)
\end{eqnarray*}
When $x > 1$, we have $h_0'(x)< 0$ and $h_0'(x) + h_1'(x) < 0$.
From which it is easy to conclude that $h'_x(x,\ep) < 0$ for any
$x > 1$ and $\ep\in (0,1)$. Detailed computations will be omitted.

Consequently for every $\ep\in(\ep_1,\ep_2)$ there is one and only
one rosette central configuration for every value of $\mu>0$. On
the other hand, our numerical study shows that,  if
$\ep\in(0,\ep_1)$ or $\ep\in(\ep_2,1)$ there is a  $\mu^*$ such
that if $0<\mu<\mu^*$ there are three rosette central
configurations and if $\mu>\mu^*$ there is only one.

In conclusion, we have the following.
\begin{proposition}
For  $n = 3$ and  $\ep\in (\ep_1,\ep_2)$, there is
exactly one rosette configuration for any $\mu > 0$. For  $n = 3$ and $\ep
\in (0,\ep_1)$ or $\ep\in (\ep_2,1)$, there exists a value
$\mu_0(\ep)>0$, such that when $\mu > \mu_0(\ep)$, $\mu =
\mu_0(\ep)$, and $ \mu < \mu_0(\ep)$, there are exactly one, two
and three rosette configurations, respectively.
\end{proposition}

%for $n=3$, there is a bifurcation in the number of classes of
%central configurations, as the parameter $\mu>0$ is varied, only
%when $\ep\in(0,\alpha_1)$ or $\ep\in(\alpha_2,1)$.
%----------------------------------------------------
%\begin{figure}
%\centerline{
%\subfigure[]{\label{fig:sub:a}\resizebox{!}{3.6cm}{\includegraphics{prova1.eps}}}
%\hspace{1cm}
%\subfigure[]{\resizebox{!}{3.6cm}{\includegraphics{prova2.eps}}}
%\caption{(a) The maximum of the function $h(x,\ep)$ as a function of $x$ for $0<\ep<1$. (b) Magnification of (a) %near $\ep=0$.}
%\label{fig:sub}
%} % caption for the whole figure
%\end{figure}
%----------------------------------------------------------
\begin{figure}
\centerline{
%\label{fig:sub:a}
\begin{tabular}{cc}
 \resizebox{!}{3.4cm}{\includegraphics{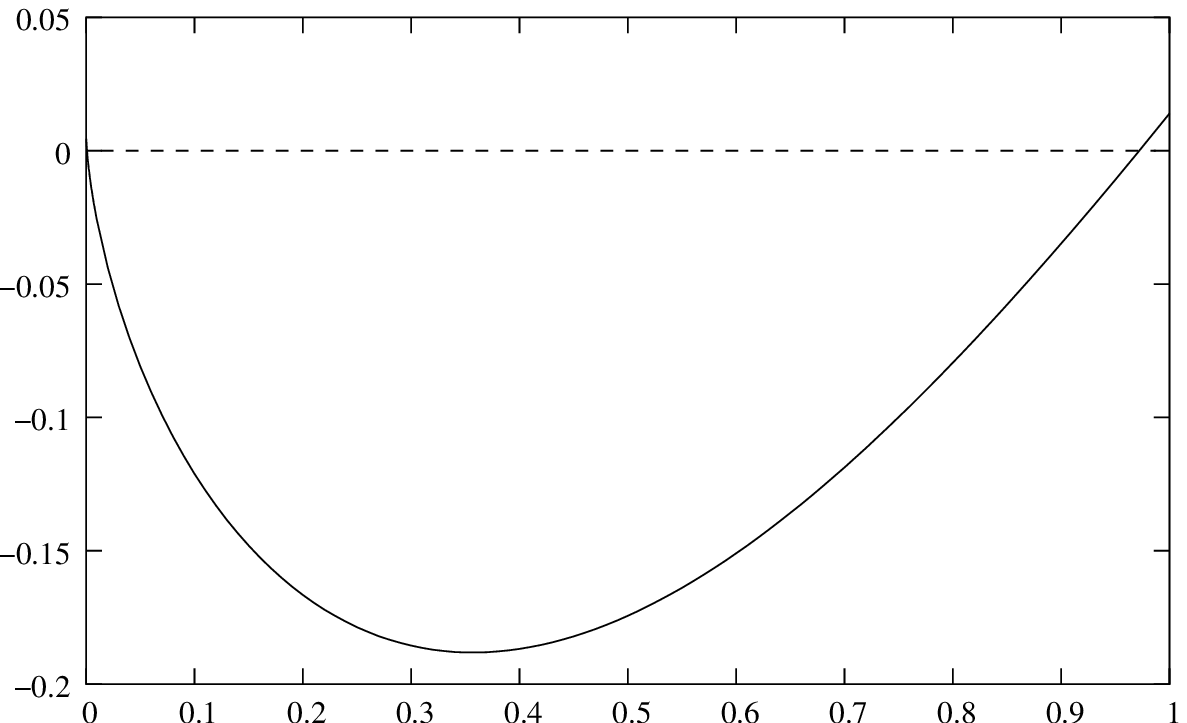}}&
 \resizebox{!}{3.4cm}{\includegraphics{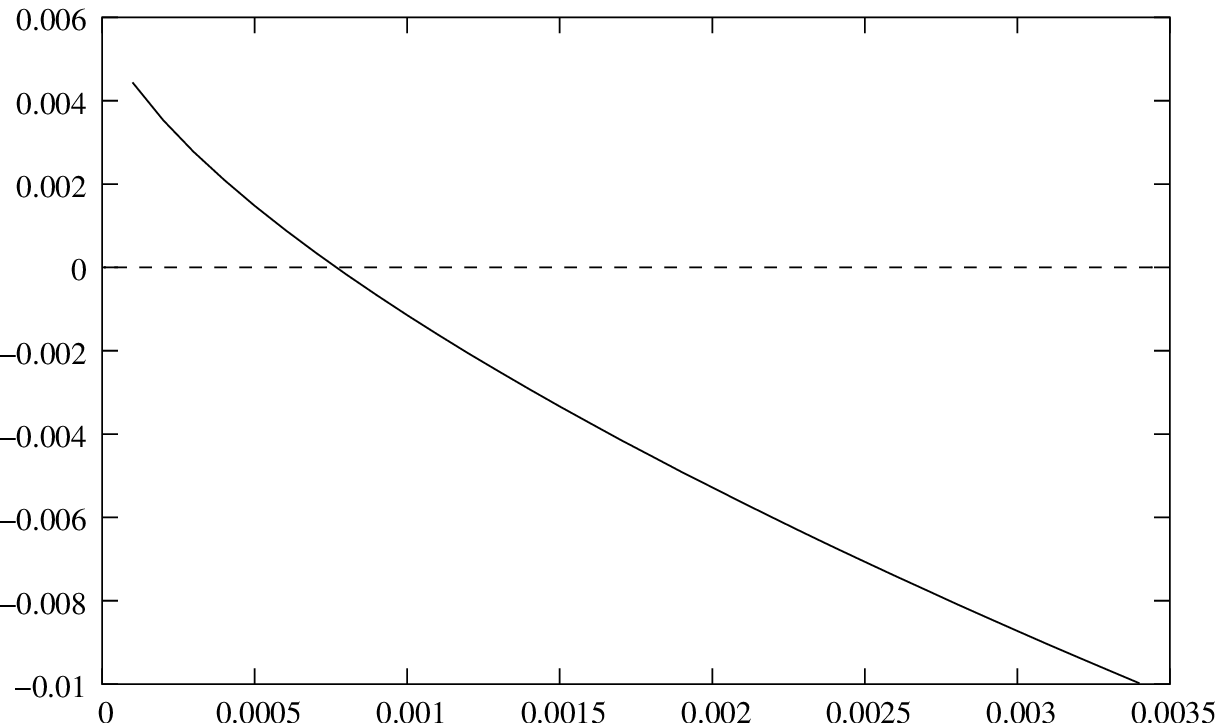}}\\
 (a)&(b)
 \end{tabular}
\caption{(a) The maximum of the function $h(x,\ep)$ as a function of $x$ for $0<\ep<1$. (b) Magnification of (a) near $\ep=0$.}
\label{fig:sub}}
\end{figure}

%%%%%%%%%%%%%%%%%%%%%%%%%%%%%%%
\acknowledgements
%%%%%%%%%%%%%%%%%%%%%%%%%%%%%%%
\noindent  MS wish to thank Giampaolo Cicogna and Donald Saari for their comments and suggestions regarding this work.

%%%%%%%%%%%%%%%%%%%%%%%%%%%%%%%%%%%%%%%%%%%%%%%%%

\end{article}
\end{document}